\documentclass[12pt]{amsart}
\usepackage{fullpage}

\usepackage{graphicx}

\usepackage{amsfonts}
\usepackage{amssymb}
\usepackage{amsmath}
\usepackage{amsthm}
\usepackage{latexsym}

\newtheorem{theorem}{Theorem}[section]

\newtheorem{proposition}[theorem]{Proposition}

\theoremstyle{plain}
\theoremstyle{plain}

\theoremstyle{remark}

\newcommand{\EE}{\mathbb{E} }
\newcommand{\1}{\mathbf{1} }
\newcommand{\PP}{\mathbb{P} }

\newcommand{\RR}{\mathbb{R} }

\newcommand{\ff}{\mathcal{F} }

\begin{document}

\author{Michael R. Tehranchi \\
University of Cambridge }
\address{Statistical Laboratory\\
Centre for Mathematical Sciences\\
Wilberforce Road\\
Cambridge CB3 0WB\\
UK}
\email{m.tehranchi@statslab.cam.ac.uk}

\date{\today}

\title{Calls, zonoids, peacocks and log-concavity}
\maketitle
 
\begin{abstract} The main results are two characterisations of log-concave densities
in terms of the collection of lift zonoids corresponding to a
peacock.  These notions are recalled and connected to  arbitrage-free
asset pricing in financial mathematics.
\end{abstract}

\section{Statement of results}
The main mathematical contributions of this short note are the following two observations.
Let $f$ be a positive and differentiable probability density function, let $F$
be the corresponding cumulative  distribution function  and let $F^{-1}$ its 
quantile function. 

\begin{theorem}\label{th:lin}
There exists a martingale $(S_t)_{t \ge 0}$ such that
$$
\inf_{K \in \RR} \left\{ 
\EE[ (S_t - K)^+ ] + p K \right\} =  \sqrt{t} f( F^{-1}(p) )  \mbox{ for all } t \ge 0, \ 0 < p < 1,
$$
if and only if $f$ is log-concave. 
\end{theorem}

\begin{theorem}\label{th:geo}
There exists a positive martingale $(S_t)_{t \ge 0}$ such that
$$
\inf_{K \in \RR} \left\{ \EE[ (S_t - K)^+ ] + p K \right\}
=  F( F^{-1}(p) + \sqrt{t} ) \mbox{ for all } t \ge 0, \ 0 < p < 1,
$$
if and only if $f$ is log-concave.
\end{theorem}

 Note
that for any increasing bijection $Y$ on the interval $[0,\infty)$, 
the process $(S_t)_{t \ge 0}$ is a martingale with respect to the filtration
$(\ff_t)_{t \ge 0}$ if and only if $(S_{Y(t)})_{t \ge 0}$ is a martingale
with respect to the filtration $(\ff_{Y(t)})_{t \ge 0}$.  In particular,
the $\sqrt{t}$ appearing in both theorems could be replaced with any increasing
bijection on $[0,\infty)$.  

The motivation for the  choice of the $\sqrt{t}$ time-parametrisation
can be found in the following statements.  We will use the notation
$$
\varphi(z) = \frac{1}{\sqrt{2\pi}} e^{-z^2/2}
$$
to denote the standard normal density function, 
$\Phi(z) = \int_{-\infty}^z \varphi(u) du$ the standard
normal cumulative distribution function and $\Phi^{-1}$
its quantile function.
  We will  let
$(W_t)_{t \ge 0}$ be a standard Brownian motion.

\begin{proposition}\label{th:B}
$$
\min_{K \in \RR} \left\{ 
\EE[ (W_t - K)^+ ] + p K \right\} =  \sqrt{t} \varphi( \Phi^{-1}(p) )  \mbox{ for all } t \ge 0, \ 0 < p < 1.
$$
\end{proposition}

\begin{proposition}\label{th:G}
$$
\min_{K \in \RR} \left\{ \EE[ (e^{W_t - t/2} - K)^+ ] + p K \right\}
=  \Phi( \Phi^{-1}(p) + \sqrt{t} ) \mbox{ for all } t \ge 0, \ 0 < p < 1.
$$
\end{proposition}

 The proof of Propositions \ref{th:B} and \ref{th:G} 
are straightforward
calculations. 
Note that the linear and geometric Brownian motions
are martingales in their common filtration.
 It was an attempt to generalise
these    propositions which lead the author to discover
the above (seemingly novel) characterisations of log-concavity in one dimension
found in Theorems \ref{th:lin} and \ref{th:geo}.

To motivate interest in
expressions of the form 
$$
C(t,K) = \EE[ (S_t-K)^+ ],
$$
where $S$ is a martingale, we appeal to the arbitrage 
theory of asset pricing.  Recall that in
a typical financial market model with zero risk-free interest and dividend rates,
there is no arbitrage if the prices of all traded assets are martingales.
  A converse statement, that the absence of arbitrage implies 
	the existence of an equivalent
	measure under which the asset prices are martingales, is true
in discrete time by the theorem of Dalang, Morton \&  Willinger \cite{DMW};
 however, it is worth noting that formulations of a `correct'
converse in continuous time are considerably more involved. 
See, for instance, the book of  Delbaen \& Schachermayer \cite{DS}
for precise details. 

Now consider a market with a stock with time $u$ price $S_u$
 and a call option written on the stock
with maturity date $t$
and strike price $K$.  In the classical arbitrage theory
recalled above, it is natural to assume that $S$ is a martingale
and that the initial price of the call is  $C(t,K)$. 

Recall that in the Bachelier model, the stock price is given 
by
$$
S_t = S_0 + \sigma W_t
$$
for constants $S_0 \in \RR$ and $\sigma > 0$ from which the call price is computed as
$$
\EE[ (S_t - K)^+ ] = \sigma \sqrt{t} \varphi\left( \frac{S_0-K}{\sigma \sqrt{t} } \right) + (S_0-K) \Phi\left( \frac{S_0-K}{\sigma \sqrt{t} } \right).
$$
Note that the linear Brownian motion in Proposition \ref{th:B}
corresponds to the Bachelier model for the stock price, with zero
initial price and unit linear volatility.

Along similar lines, recall also that in the Black--Scholes model, the 
the stock price is given 
by
$$
S_t = S_0e^{ \sigma W_t - \sigma^2 t/2 }
$$
for constants $S_0 > 0$ and $\sigma > 0$ from which the call price is computed as
$$
\EE[ (S_t - K)^+ ] = \Phi\left( \frac{\log(S_0/K)}{\sigma \sqrt{T} } + \frac{\sigma \sqrt{T} }{2} \right) - 
K  \Phi\left( \frac{\log(S_0/K)}{\sigma \sqrt{T} } - \frac{\sigma \sqrt{T} }{2} \right)
$$
for $K > 0$. 
The 
geometric Brownian motion in Proposition \ref{th:B} corresponds to the Black--Scholes
model (with unit initial price and unit geometric volatility).
We note here that Proposition \ref{th:G} appears in  \cite{SIFIN}
and is employed   to derive upper bounds on Black--Scholes
implied volatility.  See section \ref{se:upperbound} for some brief details.

Note that, in Theorems \ref{th:lin} and \ref{th:geo}, only the  marginal laws
of the random variables $S_t$ appear explicitly, but not the joint law of the 
process $(S_t)_{t \ge 0}$.  Indeed,  the filtration $(\ff_t)_{t \ge 0}$
for which the martingale property is defined is only implicit.
Therefore, we find it useful to recall the definition of a term
 popularised by Hirsh, Profeta, Roynette \& Yor \cite{HPRY}: 
a peacock is a collection of random variables $(S_t)_{t \ge 0}$ with
the property that there 
exists a filtered probability space on which a martingale  $(\tilde S_t)_{t \ge 0}$
is defined  such that $S_t \sim \tilde S_t$ for all $t \ge 0$.
The term peacock is derived from the French acronym PCOC, 
Processus Croissant pour l'Ordre Convexe. 
Peacocks have a useful characterisation in terms of 
the prices of call options 
thanks to a theorem of Kellerer \cite{kellerer}. 
\begin{theorem}
The family $(S_t)_{t \ge 0}$ of integrable random variables is a peacock if and only if
the following holds:  the map $t \mapsto \EE(S_t)$ is constant and the map
$t \mapsto \EE[ (S_t-K)^+ ]$ is increasing for all $K \in \RR$.
\end{theorem}
See the paper  \cite{HR} of Hirsh \& Roynette for a proof.

Finally, to see why one might want to compute what the Legendre transform
of a call price $C(t,K)$ with respect to the strike parameter $K$, we 
recall that the zonoid  of an integrable random
$d$-vector $X$ is the set 
$$
Z_X = \left\{ \EE[X g(X)  ]  \mbox{ measurable } g:\RR^d \to [0,1]  \right\} \subseteq \RR^d,
$$
and that the lift zonoid of $X$ is the zonoid of the $(1+d)$-vector $(1,X)$ given
by
$$
\hat Z_X =  \left\{ ( \EE[g(X)], \EE[ X g(X)   ] )  \mbox{ measurable } g:\RR^d \to [0,1]  \right\} \subseteq \RR^{1+d}.
$$

The notion of lift zonoid was introduced in
the paper of Koshevoy \& Mosler \cite{KM}.
 We note that the calculation of the lift zonoid of a 
Gaussian measure can be found in another paper of  
Koshevoy \& Mosler \cite[Example 6.3]{KM1}.  We will see that this
calculation is essentially our Proposition \ref{th:B}.

In the case $d=1$, the lift zonoid $\hat Z_{X}$ is a convex set 
contained in the rectangle
$$
[0,1] \times [-m_-, m_+ ].
$$
where $m_{\pm} = \EE( X^{\pm}).$
We can define the upper boundary of the lift zonoid 
by the function $\hat C_X: [0,1] \to \RR$ given by
\begin{align*}
\hat C_{X}(p) &= \sup\{ q: (p,q) \in \hat Z_{X} \} \\
&= \sup\left\{ \EE[ X g(X) ],   \ \   \mbox{ measurable }   g:\RR
 \to [0,1] \mbox{ with } \EE[ g(X) ] = p \right\}.
\end{align*}
Note that by replacing $g$ with $1-g$ in the definition, we see that
$\hat Z_{X}$ is symmetric about the point $(\tfrac{1}{2}, \tfrac{1}{2} m )$
where $m= m_+-m_- = \EE(X)$.
Hence, we can recover $\hat Z_{X}$ from its upper boundary from the formula
$$
\hat Z_{X}= \left\{ (p, q): 0 \le p \le 1,  m  - \hat C_X(1-p)  \le q \le \hat C_X(p) \right\}.
$$

Our interest in the notion of lift zonoid  is explained by the following result:
\begin{proposition}\label{th:duality} We have $\hat C_X(0) = 0$, $\hat C_X(1) = \EE(X)$ and 
$$
\hat{C}_X(p) = \min_{K \in \RR} [ C_X(K) + p K ]   \mbox{ for all } 0 < p < 1
$$
where $C_X(K) = \EE[ (X-K)^+ ]$.  Furthermore, we have
$$
C_X(K) = \max_{0 \le p \le 1} [ \hat C_{X}(p) - pK ]  \mbox{ for all } K \in \RR.
$$
\end{proposition}

Note that if we let 
$$
\Theta(K) = \PP(X \ge K)
$$
then we have
$$
C_X(K) = \int_K^{\infty} \Theta(\kappa) d\kappa
$$
by Fubini's theorem.  Also 
if we define the inverse function $\Theta^{-1}$ for $0 < p < 1$ by
$$
\Theta^{-1}(p) = \inf\{ K: \Theta(K) \ge p \}
$$
then by a result of Koshevoy \& Mosler  \cite[Lemma 3.1]{KM} we have
$$
\hat C_X(p) = \int_0^p \Theta^{-1}(\phi) d\phi.
$$
These representations could used to prove Proposition \ref{th:duality}.
However since the result can be viewed as an 
application of the Neyman--Pearson lemma, we include a short proof for completeness. 
\begin{proof}
For any measurable function $g$ valued in $[0,1]$ 
we have
$$
X g(X) \le  (X-K)^+ + K g(X)
$$
with equality when
$g$ is of the form
$$
g   =  \lambda \1_{(K, \infty) } +  (1-\lambda)  \1_{[K, \infty)} 
$$
where $\lambda \in [0,1]$.  Computing expectations and optimising over
$g$ yields
$$
\hat C_X(p) \le C_X(K) + Kp
$$
with equality if
$$
\PP(X > K  ) \le p \le \PP(X \ge K ).
$$
\end{proof}

We remark that the explicit connection between lift zonoids and the price of call options has
been noted before, for instance  in the paper of Mochanov \& Schmutz \cite{MS}.
 
It is interesting to observe 
that a consequence of Proposition \ref{th:duality}
is that for two integrable random variables $X$ and $Y$ with
the same mean $\EE(X) = \EE(Y)$, that
the following are equivalent, as noted by Koshevoy \& Mosler \cite[Theorem 5.2]{KM},
\begin{itemize}
\item  $X$ is dominated by $Y$ with respect to the lift zonoid order, in the sense that $\hat Z_X \subseteq \hat Z_Y$,
\item $\hat C_X(p) \le \hat C_Y(p)$ for all $0 \le p \le 1$,
\item $C_X(K) \le C_Y(K)$ for all real $K$,
\item   $X$ is dominated by  $Y$ with respect to the convex order, in the sense that
$\EE[ \psi(X) ] \le \EE[ \psi(Y) ]$ for all convex $\psi$ for which the expectations
are defined.
\end{itemize}
When $d > 1$, things are slightly more subtle. In particular, see the paper of
Koshevoy \cite{K} for an example of random vectors $X$ and $Y$ such that 
$X$ is dominated by $Y$ with respect to the lift zonoid order, and yet $X$
is not dominated by $Y$ with respect to the convex order.

We now briefly look at the case where the random $X$ is strictly positive:
\begin{proposition}\label{th:duality2} Suppose $\PP(X> 0 ) = 1$ and that $\EE(X) = m.$
Then the upper boundary of its lift zonoid is a strictly increasing continuous
function $\hat C_X: [0,1] \to [0,m]$.   Its inverse is given by $\hat C_X^{-1}(0)=0$
and $\hat C_X^{-1}(m) = 1$, and
$$
\hat C_X^{-1}(q) =   \max_{K > 0} \frac{ q  - C_X(K)}{K}    \mbox{ for all } 0 < q < m.
$$
\end{proposition}

Since Proposition \ref{th:duality2} appears to be new, or at least its statement
 does not seem to be easy to find in the literature, we now offer a proof.

\begin{proof}  Since   $\PP( X > 0) = 1$ almost surely we can conclude
$$
\Theta^{-1}(p) > 0
$$
for all $0 < p < 1$.  This shows that  $\hat C_X$ is strictly increasing.

Now, let $Y$ be a positive random variable such that
$$
\PP( Y \le K ) = \EE\left[ \frac{X}{m} \1_{\{ Y \ge 1/K \}} \right] \mbox{ for all } K > 0.
$$
That is to say, the distribution of $Y$ is given by  the distribution of $1/X$
under the equivalent measure with density $X/m$.   Note hat
\begin{align*}
\hat C_X^{-1}(q) & = \inf\{ \EE[ g(X) ], \ \mbox{ measurable } g:\RR \to [0,1]
 \mbox{ with } \EE[ X g(X) ] = q \} \\
& = 1 -  \sup\{ \EE[ g(X) ], \ \mbox{ measurable } g:\RR \to [0,1]
 \mbox{ with } \EE[ X g(X) ] = m-q \} \\
& = 1 - \sup\left\{ m \EE\left[ Y g(Y) \right], \ \mbox{ measurable } g:\RR \to [0,1] \mbox{ with } 
\EE\left[   g(Y) \right] = 1- q/m  \right\} \\
& = 1 - m  \hat C_Y(1 -q/m) \\ 
& =  1 - \min_{K} [ m C_Y(K) + (m-q) K ], 
\end{align*} 
where the minimisation can be restricted to positive $K$.  
The proof is concluded by noting that
\begin{align*}
m C_Y(K)  &= \EE[ (1- XK)^+ ] \\
& = \EE[ 1- XK + (XK-1)^+ ] \\
& = 1 - Km + K C_X(1/K)
\end{align*}
for $K > 0$.
\end{proof}

To prove Theorems \ref{th:lin} and \ref{th:geo}, we now need to know how to characterise the call
price function $C_X(\cdot)$ of an integrable random variable $X$,
as well as  the upper boundary $\hat C_X(\cdot)$ of its lift zonoid.
The following fact is well known.  
The proof is well-known, and can be found in the paper of 
Hirsh \& Roynette \cite[Proposition 2.1]{HR}, for instance.  In the financial context, the 
 link from the call price function $C_X$ to and the random variable $X$ is sometimes called the 
Breeden--Litzenberger formula.

\begin{proposition}\label{th:call}
Suppose that the function $C:\RR \to \RR_+$ is decreasing, convex and
satisfies
$$
C (K) \to 0 \mbox{ as } K \to \infty
$$
and 
$$
C(K) + K \to m \mbox{ as } K \to - \infty
$$
for some finite constant $m$.  There there exists 
a (unique in law)  integrable random variable $X$  such that
$$
m =\EE(X)
$$
and 
$$
C(K) = C_{X}(K) \mbox{ for all } K \in \RR.
$$
\end{proposition}

The next result is the lift zonoid version of Proposition \ref{th:call}.
Its proof can be found in the paper of Koshevoy \& Mosler
\cite[Theorem 3.5]{KM}.
\begin{proposition}
Suppose that $\hat C:[0,1] \to \RR$ is concave function
and such that $\hat C(0) = 0$ and $\hat C(1) = m$.
Then there exists a (unique in law) integrable random variable $X$ such that
$$
m =\EE(X)
$$
and 
$$
\hat C(p) = \hat C_X(p)  \mbox{ for all } 0 \le p \le 1.
$$
\end{proposition}

In light of the proceeding discussion, we now see that Theorem \ref{th:lin}
is equivalent to 
\begin{proposition}\label{th:lin1}
The map
$$
p \mapsto f( F^{-1}(p) )
$$
is concave if and only if $f$ is log-concave.
\end{proposition}\label{th:geo1}
while Theorem \ref{th:geo} is equivalent to
\begin{proposition}
The map
$$
p \mapsto F( F^{-1}(p) + y)
$$
is concave for all $y \ge 0$ if and only if $f$ is log-concave.
\end{proposition}

\begin{proof}[Proof of Propositions \ref{th:lin1} and \ref{th:geo1}]
First, let
$$
G(p) = f(F^{-1}(p)).
$$
Note that the derivative  is given by the formula
$$
G'(p) = \frac{f'(a)}{f(a)}
$$
where $a = F^{-1}(p)$.  Therefore the function
 $G'$ is decreasing if and only $\log f$ is concave.

Now fix $y \ge 0$ and let
$$
H_y(p) =  F( F^{-1} (p) + y ).
$$
Note that 
$$
H'_y(p) = \frac{ f( F^{-1}(p) + y)}{f(F^{-1}(p))}.
$$
Therefore, the function $H'_y$   is decreasing if and only if
$$
 \log f(b+y) - \log f(b)  \le  \log f(a+y) - \log f(a)
$$
for all $b \ge a$. 
 This last condition holds for all $y \ge 0$ if and only if $\log f$ is concave.
\end{proof}

\section{Various remarks}
We conclude this note with various remarks expanding on the main results.
As before, let $f$ be a strictly positive  probability density, and 
$F$ its cumulative distribution function  and $F^{-1}$ its quantile function. 
In this section, we further assume that $f$ is log-concave.
We will use the notation
$$
G = f \circ F^{-1}
$$
and
$$
H_y = F( F^{-1}(\cdot) + y).
$$

\subsection{Group property}
To better understand the connection between the functions $G$ and $H_y$
 introduced in the proofs of Propositions \ref{th:lin1} and \ref{th:geo1}, 
 note that
the family of functions $(H_y)_{y \in \RR}$ on $[0,1]$ form
a group with respect to composition
$$
H_{y_1+y_2} = H_{y_1} \circ H_{y_2}.
$$
Note that $H_0 = \mathrm{Id}$, and that the
 generator of this group is given by $G$ in the sense that
$$
\frac{H_y - \mathrm{Id}}{y} \to G
$$
as $y \to 0$.  The above observations appear in the paper
of Kulik \& Tymoshkevych \cite{kt} in the case where
 $f= \varphi$ is the standard normal density.  In this
case, the function $G = \varphi \circ \Phi^{-1}$ is called
the Gaussian isoperimetric function.

\subsection{Recovering $F$}
Given the function
$G$  we can 
solve for the distribution function $F$ and hence the density $f$.
Moreover, the solution is unique up to a free location parameter.
Indeed, fix $p_0$ and declare $F(a) = p_0$.  Then
$$
\int_{p_0}^p \frac{d q}{G(q) } = \int_{p_0}^p \frac{ dq}{f (F^{-1}(q))}
= F^{-1}(p) - a
$$
from which $F$ can be recovered.

Furthermore, given the family of functions $(H_y)_{ y \ge 0}$
we can recover $F$ in two different ways, where again we 
fix $p_0$ and set $F(a) = p_0$.  Firstly, note that
$$
\partial_y H_y(p) |_{y=0} = G(p)
$$
to recover $F$ as described above.
 Secondly, we can simply observe that
$$
F(x) = H_{x-a}(p_0) \mbox{ for all } x \in \RR.
$$

\subsection{Symmetries}  
If the integrable random variable $X$ has  arithmetic  symmetry, in the sense that $-X$
has the same law as $X$, then its call function satisfies
\begin{align*}
C_X(K) & = \EE[ (X-K)^+ ] \\
& =   \EE[ X - K + (-X+K)^+ ] \\
& =  - K + C_X(-K).
\end{align*}
The upper boundary of its lift zonoid satisfies
\begin{align*}
\hat C_X(p) & = \min_K [ C_X(K) + pK ]
\\& =  \min_K [ C_X(-K)  - (1-p) K ] \\
& = \hat C_X(1-p).
\end{align*}

Note that if $f$ is an even function and $\hat C_X(p) = f( F^{-1}(p) )$, then 
$X$ has arithmetic symmetry since in this case $F^{-1}(1-p) = - F^{-1}(p)$.

A strictly positive random variable $X$ has geometric symmetry if
$$
\EE[ \psi(X) ] = \EE\left[ X \psi\left( \frac{1}{X} \right) \right]
$$
for all non-negative $\psi$.  In particular, geometric symmetry implies
$\EE(X) =  1$ and that the call function satisfies the put-call symmetry
formula
\begin{align*}
C_X(K) & = \EE[ (X-K)^+ ] \\
&  =   \EE[ X - K + (K-X)^+ ] \\
& = 1 - K + K \EE\left[ X \left( \frac{1}{X} - \frac{1}{K} \right)^+ \right] \\
& = 1  - K + K C_X(1/K)
\end{align*}
for $K > 0$. In this case, the upper boundary of its lift zonoid satisfies
\begin{align*}
\hat C_X(p) & = \min_K [ C_X(K) + pK ]  \\
& = 1 -  \max_{K > 0} K (1 - p -  C_X(1/K))  \\
& = 1 - \hat C_X^{-1}(1-p)
\end{align*}
by Proposition \ref{th:duality2}.  Since the lift zonoid of $X$ is given by
$$
\hat Z_X = \{ (p,q): \ 0 \le p \le 1, 1 - \hat C_X(1-p) \le q \le \hat C_X(p) \}
$$
we see that another way to characterise geometric symmetry of $X$ is that the 
lift zonoid is symmetric about the line $p=q$.

Note that if $f$ is even and $\hat C_X(p) = F( F^{-1}(p) + y)$, then 
$X$ has geometric symmetry thanks to the calculation 
$$
\hat C_X^{-1}(q) = F( F^{-1}(q) - y) = 1 - \hat C_X(1-q).
$$

Applications of arithmetic and geometric symmetries to construct
semi-static hedging strategies for certain barrier options
is explored in the paper of Carr \& Lee \cite{CL}.

\subsection{The initial stock price}
Note that $G(1) = 0$.   Hence, if the martingale $S$ is such that 
$$
\min_{K}\{ \EE[(S_t-K)^+ + pK \} = Y(t) G(p)
$$
for some increasing function $Y$, then $\EE(S_t) = S_0 = 0$.    To consider models with non-zero
initial prices, let $\tilde S_t = s + S_t$ for some constant $s$.  Then
\begin{align*}
\min_{K}\{ \EE[(\tilde S_t-K)^+ + pK \} &= 
\min_{K}\{ \EE[(S_t-(K-s) )^+ + p(K-s) + ps \} \\
& = ps + Y(t) G(p).
\end{align*}

Similarly, note that $H_y(1) = 1$ for all $y \ge 0$.  Hence if
$$
\min_{K}\{ \EE[(S_t-K)^+ + pK \} = H_{Y(t)}(p)
$$
then $S_0 = 1$.   To consider more general initial prices, let $s > 0$ and 
$\tilde S_t = s S_t$.  Then
\begin{align*}
\min_{K}\{ \EE[(\tilde S_t-K)^+ + pK \} &= 
s \min_{K}\{ \EE[(S_t-K/s )^+ + pK/s  \} \\
& = s H_{Y(t)}(p).
\end{align*}

\subsection{The call function}
Given the upper boundary of the lift zonoide $\hat C_X$, we can compute
the corresponding call function $C_X$ by Proposition \ref{th:duality}.
We now explore these representations when $\hat C_X$ has the specific
forms appearing in Theorems \ref{th:lin} and \ref{th:geo}.

First suppose
\begin{align*}
\hat C_X(p)  & = sp + G(p) \\ & =sp+  f( F^{-1}(p) ) 
\end{align*}
for some constant $s$. 
Our first calculation is then
\begin{align*}
C_X(K) & =  \max_{0 \le p \le 1 }[ \hat C_X(p) - pK]   \\
  & = f ( U(K-s) ) - F( U(K-s) ) (K-s),
\end{align*}
where $U$ is inverse of the decreasing function $f'/f = (\log f)'$.
Furthermore,  we have  the calculation 
\begin{align*}
\PP\left( X \ge K  \right) &= -  C_X'(K) \\
& =  F(U(K-s) ).
\end{align*}
In the special case when $f= \varphi$ is the standard normal density, we have $U(x) = - x$ and we
see that $X$ has the standard 
normal distribution in agreement with the Bachelier model and Proposition \ref{th:B}.

Similarly, if 
\begin{align*}
\hat C_X(p)  & = s H_y(p) \\ & = s F( F^{-1}(p) + y) 
\end{align*}
for some $y > 0$ and $s > 0$, then
we have
\begin{align*}
C_X(K) =   F( V_y(K/s) + y ) - F( V_y(K/s) ) K,
\end{align*}
where $V_y$ is the inverse of the decreasing function $f(\cdot + y)/f$. Furthermore,
 we have
$$
\PP\left( X \ge K  \right) = F(V_y(K/s) ).
$$
In the special case when $f= \varphi$ and $V_y(x)= -\log x/y - y/2$,   we
see that $ \log X$ has the 
normal distribution with mean $-y^2/2$ and variance $y^2$,
 in agreement with the Black--Scholes model and Proposition \ref{th:G}.

\subsection{Implied volatility}\label{se:upperbound}
Let
$$
C_F(y, K)   =  \max_{0 \le p \le 1} [H_y(p) - pK].
$$
This corresponds to a  call price on a stock with initial price $S_0 = 1$ and strike $K$, or equivalently
 the  call price normalised by the initial stock price and $K$ is the strike price
normalised by the initial stock price.

We now show   for fixed $K$ that $C_F(\cdot,K)$ takes values in 
the interval $[(1-K)^+, 1)$.  Note that $H_0(p) = p$ and so
\begin{align*}
C_F(0,K) &= \max_{0 \le p \le 1} (1-K)p  \\
& = (1-K)^+.
\end{align*}
Also, note for $y > 0$ that we have
$$
C_F(y,K) = F( V_y(K) + y ) - F( V_y(K) ) K.
$$
where 
$V_y$ is the inverse of the decreasing function $f(\cdot + y)/f$.
  It is clear from the 
formula that $C_F(y, K) < 1$.

Now, one can verify by differentiation that 
$$
C_F(y,K) = (1-K)^+ + \int_0^y  f( V_u(K) + u ) du.
$$
Since the quantity $y$ corresponds to $\sigma \sqrt{t}$ in the Black--Scholes model,
the above formula can be seen as a generalisation of the formula for the vega,
the sensitivity of the call price with respect to the Black--Scholes volatility.
In particular, we see that 
$$
y \mapsto C_F(y,K)
$$
is continuous and strictly increasing.  We have
for every $0 < p < 1$ the inequality
\begin{align*}
C_F(y,K)  \ge F( F^{-1}(p ) +y ) - p  K.
\end{align*}
By taking   $y \uparrow \infty$ and then $p  \downarrow 0$ we see that $C_F(y,K) \to 1$. 
In particular, for every $c \in [(1-K)^+, 1)$ there
is a unique $y$ such that 
$$
c =  C_F(y^*,K).
$$ 
This $y^*$  generalises the notion of Black--Scholes implied volatility.

We now show that $y^*$ can be recovered from $c$ by the formula
$$ 
y^* = \min_{0 \le p \le 1}[ F^{-1}(c+ pK) -F^{-1}(p) ].
$$ 
Indeed, we can rearrange the inequality
\begin{align*}
c & \le  F( F^{-1}(p) +y^* ) - p K,
\end{align*}
which holds for all $0 \le p \le 1$, to yield the bound
$$
y^*  \le  F^{-1}(c+ pK) -F^{-1}(p).
$$
Since there is equality above when $p = F( V_y(K) )$, the claim 
is proven.   

The above representation of implied volatility
as the value of a minimisation problem  
was exploited in \cite{SIFIN}
to obtain upper bounds on  Black--Scholes implied volatility.

\subsection{Local volatility}
Suppose the martingale $S$ evolves according to
the stochastic differential equation
$$
dS_t = \sigma(t,S_t) dW_t
$$
where $W$ is a Brownian motion, and the local volatility 
function $\sigma:\RR_+ \times \RR \to \RR_+$
is continuous.  It is well known that the linear volatility 
$\sigma$ can be recovered from the 
call prices 
$$
C(t,K) = \EE[ (S_t-K)^+]
$$
by Dupire's formula
$$
\sigma(t,K)^2 =  \frac{ 2 \partial_t C}{\partial_{KK} C} \big|_{(t,K)}.
$$
See, for instance, the book of Musiela \& Rutkowski \cite[Proposition 7.3.1]{MR} for a 
precise statement and proof.

Hence by Proposition \ref{th:duality}, the function $\sigma$ can be recovered
from the upper boundary  of the lift-zonoid
$$
\hat C(t,p) = \sup\{ \EE[ S_t g(S_t) ],  \ \ \mbox{ measurable }g:\RR \to [0,1] \mbox{ with  }\EE[ g(S_t) ] = p \}
$$
 via
$$
\sigma(t, \partial_p \hat C|_{(t,p)})^2 =  - 2 \partial_t \hat{C} \ \partial_{pp} \hat C|_{(t,p)}.
$$

In the case where $f$ is log-concave and 
$$
\hat C(t,p) = S_0 p + Y(t) f( F^{-1}(p) )
$$
for an increasing function $Y$, 
we get
$$
\sigma(t, K)^2 = -2 Y(t) \dot{Y}(t) ( \log f)''\left[ U\left(\frac{K-S_0}{Y(t)} \right) \right]
$$
where $U$ is the inverse of the decreasing function $f'/f = (\log f)'$.
Note that when we specialise to 
$Y(t)= \sigma  \sqrt{t}$ and $f = \varphi$ the standard normal density, the right side
is equal to the constant $\sigma^2$, in agreement with the Bachelier model and Proposition \ref{th:B}.

Finally, in the case where $f$ is log-concave and 
$$
\hat C(t,p) = S_0 F( F^{-1}(p) + Y(t) )
$$
we get
$$
\bar \sigma(t, K)^2 = 2 \dot{Y}(t) \left( (\log f)'(V_{Y(t)}(K/S_0) ) -  (\log f)'(V_{Y(t)}(K/S_0)+ Y(t)) \right)
$$
where $\bar \sigma(t,K) = \sigma(t,K)/K$ for $K > 0$ is the geometric
volatility, and 
$V_y$ is the inverse of the decreasing function $f( \cdot +y )/f$ for $y > 0$.
Again, when we specialise  to the case
$Y(t)= \sigma \sqrt{t}$ and $f = \varphi$ we see that the right side is the constant $\sigma^2$, in 
line with the Black--Scholes model and Proposition \ref{th:G}.

\section{Acknowledgement} 
I would like to thank the Cambridge Endowment for Research in Finance for 
their support. I would also like to thanks Thorsten Rheinl\"{a}nder for introducing
me the notion of a lift zonoid.

\end{document}